\documentclass{amsart}
\usepackage{epsfig}

\newtheorem{theorem}{Theorem}
\newtheorem{lemma}[theorem]{Lemma}

\usepackage{geometry} 
\usepackage{amsmath}
\usepackage{graphicx}
\usepackage{epsfig}
\usepackage{latexsym}
\usepackage{epsfig}
\newcommand{\PP}{{\mathbb P}}

\newcommand{\EE}{{\mathbb E}}

\begin{document}

\title{A consistency lemma in statistical phylogenetics}
\author{Mike Steel}
\address{Biomathematics Research Centre, University of Canterbury, Christchurch, New Zealand.}
\email{mike.steel@canterbury.ac.nz}
\date{January 26, 2015}

\begin{abstract}
This short note provides a simple formal proof of a folklore result in statistical phylogenetics concerning the convergence of bootstrap support for a tree and its edges.
\end{abstract}

\maketitle

\section{Definitions and preliminaries}
In this note $T$ will refer to any rooted or unrooted phylogenetic tree, and $T^{-\rho}$ will refer to the unrooted tree
obtained from $T$ by suppressing the root vertex $\rho$ if it has one (i.e.  if $T$ is unrooted then $T^{-\rho} = T$).   Let $\theta$ be a vector of continuous parameters -- including the branch lengths of $T$, along with possibly other continuous parameters required to specify a model of character evolution on $T$.  Let $\Theta$ denote the set of values $\theta$ may take. Branch lengths, in particular, are assumed to be strictly positive and finite; and  in general $\Theta$ will be some open subset of Euclidean space. 
Consider any stochastic process (e.g. Markov process, or mixture of Markov processes) which assigns to each pair $(T, \theta)$ 
a probability distribution ${\bf s} ={\bf s}(T ,\theta)$ on discrete, finite-state characters at the tips of the tree.  We assume throughout that the map $\theta \mapsto {\bf s}(T ,\theta)$ is continuous. Such models are central to statistical phylogenetics and methods for reconstructing phylogenetic trees from aligned genetic (e.g. DNA) sequences. 
 A {\em tree reconstruction method}  $\psi$ is any method that reconstructs a set of one or more  unrooted phylogenetic trees from any given distribution $\hat{\bf f}$ of site pattern frequencies.
Suppose we generate $k$ sites i.i.d. from $(T, \theta)$, and let $\hat{\bf s}$ be the random variable equal to the resulting proportion of  site patterns (character types). The method $\psi$ is a {\em statistically consistent} estimator of the unrooted topology of $T$  if the
probability that $\psi(\hat{\bf s}) = \{T^{-\rho}\}$ converges to 1 as $k \rightarrow \infty$\footnote{There is a slightly stronger definition involving almost sure convergence rather than convergence in probability, and the results here can be extended to that setting also.}. Suppose that $\psi$ satisfies the following condition:
\begin{itemize}
\item[(*)]For every tree $T$ for which $T^{-\rho}$ is fully-resolved (i.e. binary), and each $\theta \in \Theta(T)$ a value $\epsilon = \epsilon_{(T, \theta)}>0$ exists for which the following inequality holds for every probability distribution $\hat{{\bf f}}$ on site patterns:
$\|\hat{\bf f} - {\bf s}(T, \theta)\|< \epsilon \Rightarrow \psi(\hat{\bf f}) =\{T^{-\rho}\}.$
\end{itemize}
Here $\| \cdot \|$ denotes any of the usual norms in Euclidean space.
Condition (*) implies the statistical consistency of $\psi$ for inferring $T^{-\rho}$ since the i.i.d. assumption ensures that  $\hat{\bf s}$ converges in probability to ${\bf s}(T,\theta)$ as $k$ grows, and so:
$$\PP(\psi(\hat{\bf s}) = \{T^{-\rho}\}) \geq \PP(\|\hat{\bf s} - {\bf s}(T, \theta))\|< \epsilon_{(T, \theta)}) \rightarrow 1, \mbox{ as $k \rightarrow \infty$}.$$
Not only does  condition (*) imply that $\psi({\bf s}(T, \theta)) = \{T^{-\rho}\}$  whenever $T^{-\rho}$ is fully-resolved but (*) also implies the stronger condition that 
for any tree $T'$ that has a different unrooted topology (fully-resolved or non-fully-resolved) from the fully-resolved tree $T$ we have:
\begin{equation}
\label{touching}
\inf_{\theta' \in \Theta(T')}\|{\bf s}(T, \theta)-{\bf s}(T', \theta')\|>0,
\end{equation}
a strong `identifiablity' condition,  referred to as `no touching' in \cite{sin}.

Condition (*) is a type of local stability condition.  It applies, for example, to distance-based tree reconstruction applied to (statistically consistent) `corrected distances' derived from the characters,
provided that the distance-reconstruction method has a positive `safety radius', which holds for many (but not all) distance-based methods, including the popular Neighbor-Joining method \cite{att}. Condition (*) also applies to MLE (maximum likelihood estimation) for models which satisfy (\ref{touching}) -- such models include the general time-reversible (GTR) Markov processes and its submodels (e.g. Jukes-Cantor type models) and certain extensions of these models. Here MLE treats $\theta$ as `nuisance parameters' to be optimized
as part of the search for the MLE tree;  given a vector $\hat{{\bf f}}$ as input, MLE selects the tree(s) $T'$ maximizing
$\sup_{\theta \in \Theta(T')} \PP(\hat{\bf f}|{\bf s}(T',\theta)).$
The proof that Condition (*) holds for models satisfying (\ref{touching}) follows from standard analytic arguments based on the continuity of the map $\theta \mapsto \PP(\hat{{\bf f}}|{\bf s}(T', \theta))$ (see e.g. \cite{cha} or \cite{sin}).

\section{Result}
Given $\hat{\bf s}$ derived from $k$ i.i.d. site patterns, let $\hat{\bf s}^*$ denote the frequency of site patterns obtained by taking an i.i.d. sample of $k$ site patterns using probability distribution $\hat{\bf s}$. Thus $\hat{\bf s}^*$ is the distribution of site patterns in a  bootstrap sample from the original data. The {\em bootstrap support of an edge $e$} of an unrooted phylogenetic tree $T'$,  is the expected proportion of such bootstrap samples for which a tree, sampled uniformly at random from  $\psi(\hat{\bf s}^*)$, has an edge that induces the same split of the leaf taxa as $e$ does in $T'$ (it is a random variable by its dependence on ${\hat{\bf s}}$, and since $\psi$ can return more than one tree).  The {\em bootstrap support for $T'$} is the random variable $\PP(\psi(\hat{\bf s}^*) = \{T'\}|\hat{\bf s})$, the expected proportion of bootstrap samples for which $\psi$ returns the single tree $T'$. The following result was motivated by a question from T. Warnow (pers. comm.). 
\begin{lemma}
Suppose $k$ sites are generated i.i.d. by ${\bf s}(T, \theta)$. Under the sufficient condition (*)  for statistical consistency, the bootstrap support of every edge $e$ of $T^{-\rho}$ converges in probability to 1 as $k \rightarrow \infty$.  Moreover, the bootstrap support for $T^{-\rho}$  converges in probability to 1 as $k\rightarrow \infty$. 
\end{lemma}
\begin{proof}
Clearly it suffices to prove the second assertion in the lemma, since, by definition,
the bootstrap support for any edge $e$ of $T^{-\rho}$ is at least $\PP(\psi(\hat{\bf s}^*) = \{T^{-\rho}\}|\hat{{\bf s}}).$
Let $X= X(\hat{\bf s})$ be the 0/1 random variable which takes the value 1 precisely if $\psi(\hat{\bf s}^*)= \{T^{-\rho}\}$, and which is 0 otherwise. 
Let $Y$ denote the expected bootstrap support for $T^{-\rho}$ given $\hat{{\bf s}}$; thus  $Y= \PP(\psi(\hat{\bf s}^*) = \{T^{-\rho}\}|\hat{\bf s}) = \EE[X|\hat{\bf s}]$ (i.e. the conditional expectation of $X$ given $\hat{\bf s}$). Notice that:
\begin{equation}
\label{ee}
\EE[Y] = \EE[E[X|\hat{\bf s}]] = \EE[X] = \PP(\psi(\hat{\bf s}^*) = \{T^{-\rho}\}).
\end{equation}
Now, as $k$ grows, $\hat{\bf s}\xrightarrow{p} {\bf s}$, and $\hat{\bf s}^*- \hat{\bf s} \xrightarrow{p} {\bf 0}$; thus
$\hat{\bf s}^*  \xrightarrow{p} {\bf s}$.  Consequently, by Condition (*),  $\PP(\psi(\hat{\bf s}^*) = \{T^{-\rho}\})$ converges to 1 as $k \rightarrow \infty$, and so, by (\ref{ee}), $\lim_{k \rightarrow \infty} \EE[Y] = 1$.
Finally, since $Y$ takes values in the interval $[0,1]$, and the expected value of $Y$ converges to 1 as $k \rightarrow \infty$, it follows that (for the bootstrap support for $T^{-\rho}$) we have $Y\xrightarrow{p}  1$ as $k \rightarrow \infty$, as required.
\end{proof}
Note that the empirical bootstrap support for an edge (or for a  tree)  given $\hat{\bf s}$,   converges in probability to the (expected) bootstrap support value defined here, as the number $N$ of independent bootstrap replicates  becomes large; hence our results are also relevant for empirical bootstrap support for large $N$.

\end{document}